\documentclass[12pt]{article}\date{}

\usepackage{amsmath,amssymb,latexsym,xspace}
\usepackage[standard]{ntheorem}
\usepackage{changebar}
\usepackage{lineno}
\ifx\xyloaded\undefined \input xy \fi

\xyrequire{arrow}
\xyrequire{curve}
\xyrequire{frame}

\newdimen\automaunit\automaunit=15pt
\newcount\stateidx   
\def\automaidxbase{0}
\let\statelabel\relax
\let\automamore\relax

\def\itera#1[#2;#3]{\ifx;#2;\else#1#2,,,;\itera#1[#3]\fi}
\def\XpandAftrGroup#1#2{\edef\GrpAftr{#2}\expandafter#1\GrpAftr}

\def\automaS #1,#2,#3,#4;{%
\XpandAftrGroup\POS%
{(#1,#2)*++[o]+=<\automaunit>[o][F-]{\statelabel{#3}}="\the\stateidx"}%
\advance\stateidx by 1%
}

\def\automaT #1,#2,#3,#4,#5;{%
\ifnum#1=#2
\XpandAftrGroup\POS{"#1";p+/\ifx,#4,u\else#4\fi\automaunit/:p}%
\ar@`{(1,-1),(2.5,0),(1,1)}_*{#3}"#1"%
\else
\POS"#1"\XpandAftrGroup%
\ar{\ifx,#4,^\else@/^#4mm/\ifdim#4mm>0mm^\else_\fi\fi*{#3}"#2"}%
\fi%
}

\def\automaI #1,#2,#3;{%
\ifx,#2,\automaI #1,l,;\else%
\POS"#1";p+/#2/**{}?<+/#2\automaunit/\ar"#1"\fi%
}

\def\automaF #1,#2;{%
\POS"#1"*\frm{ee}%
}

\def\automa[#1][#2][#3][#4]{%
\xy 0;<4\automaunit,0cm>:%
\stateidx=\automaidxbase%
\save%
\itera\automaS[#1;;]%
\itera\automaT[#2;;]%
\itera\automaI[#3;;]%
\itera\automaF[#4;;]%
\restore%
\automamore\endxy}

\theorembodyfont{\upshape}
\newtheorem{conjecture}{Conjecture}

\newcommand{\qed}{\null\hfill$\Box$}

\newcommand{\fact}{\operatorname{Fact}}
\newcommand{\card}{\operatorname{Card}}

\newcommand{\A}{\ensuremath{\mathcal A}\xspace }

\newcommand{\cM}{\mathcal{M}}
\newcommand{\cF}{\mathcal{F}}

\title{On incomplete and synchronizing finite sets
\thanks{This work was partly supported by MIUR project PRIN 2010/2011
{``Automi e Linguaggi Formali: Aspetti Matematici e Applicativi''.}}}

\author{Arturo Carpi \\
Dipartimento di Matematica e Informatica, \\
Universit\`a degli Studi di Perugia, \\
via Vanvitelli 1, 06123 Perugia, Italy. \\
e-mail: carpi@dmi.unipg.it
\and Flavio D'Alessandro \\
Dipartimento di Matematica, \\
La Sapienza Universit\`a di Roma \\
Piazzale Aldo Moro 2, 00185 Roma, Italy. \\
e-mail: dalessan@mat.uniroma1.it}

\begin{document}
\maketitle

\begin{abstract}
  This paper situates itself in the theory of variable length codes and of finite automata where the concepts of completeness and synchronization play a central role.
  In this theoretical setting, we investigate the problem of finding upper bounds to the minimal length of synchronizing words and incompletable words of a finite language $X$ in terms of the length of the words of $X$.
  This problem is related to two well- known conjectures formulated by \v Cern\'y and Restivo, respectively.
  In particular, if Restivo's conjecture is true, our main result provides a quadratic bound for the minimal length of a synchronizing pair of any finite synchronizing complete code with respect to the maximal length of its words.
  
\end{abstract}

\section{Introduction}
The concepts of completeness and synchronization play a central role in Formal Language Theory since they appear in the study of several problems on variable length codes and on finite automata \cite{BP-book}.
According to a well-known result of Sch\"u\-tzen\-ber\-ger, the property of completeness provides an algebraic characterization of finite maximal codes, which are the objects used in Information Theory to construct optimal sequential codings.

Let $X$ be a set of words on an alphabet $A$ and let $X^*$ be its Kleene closure.
The set $X$ is {\em complete} if any word on the alphabet $A$ is a factor of some word belonging to $X^*$, otherwise it is {\em incomplete}.
In the latter case, any word which is factor of no word of $X^*$ is said to be {\em incompletable in $X$}.

In \cite{Restivo}, Restivo conjectured that a finite incomplete set $X$ has always an incompletable word whose length is quadratically bounded by the maximal length of the words of $X$.
Results on this problem have been obtained in \cite{boe-deluca-restivo,Fici,gusev,Restivo}.
The property of synchronization plays a natural role in Information Theory where it is relevant for the construction of decoders that are able to efficiently cope with decoding errors caused by noise during the data transmission.
A set $X$ is {\em synchronizing} if there are two words $u, v$ of $X^*$ such that whenever $ruvs \in X^*$, $r,s \in A^*$, one has also $ru, vs \in X^*$.
The pair of words $(u,v)$ is called a {\em synchronizing pair of $X$}.

In the study of synchronizing sets, the search for synchronizing words of minimal length in a prefix complete code is tightly related to that of synchronizing words of minimal length for synchronizing complete deterministic automata and the celebrated \v Cern\'y Conjecture \cite{Cerny} (see also \cite{A,BP,BBP,Carpi,CD1,Acta,DLT,DLT10,CD3,Cerny,Pin78,Pin78bis,V08} for some results on the problem).
In particular, in \cite{BP} (see also \cite{BBP}), B\'eal and Perrin have proved that a complete synchronizing prefix code $X$ on an alphabet of $d$ letters with $n$ code-words has a synchronizing word of length $O(n^2).$

In this paper we are interested in finding upper bounds to the minimal lengths of incompletable and synchronizing words of a finite set $X$ in terms of the size of $X$.

We recall that the {size of $X$} is the parameter $\ell({X})$ defined as the maximal length of the words of $X$.

Let $\cal L$ be a class of finite languages.
For all $n,d>0$, we denote by $R_{\mathcal L} (n,d)$ the least positive integer $r$ satisfying the following condition: any incomplete set $X\in {\cal L}$ on a $d$-letter alphabet such that $\ell(X)\leq n$ has an incompletable word of length $r$.
Similarly, we denote by $C_{\cal L} (n,d)$ the least positive integer $c$ satisfying the following condition: any synchronizing set $X\in {\cal L}$ on a $d$-letter alphabet such that $\ell(X)\leq n$ has a synchronizing pair $(u,v)$ such that $|uv|\leq c$.

In this context, the main result of this paper provides a bridge between the parameters $R_{\cal L} (n,d)$ and $C_{\cal L} (n,d)$.
More precisely, denoting by ${\cal F}$ and by ${\cal M}$ the classes of finite languages and of complete finite codes respectively, we show that, for all $n,d>0$,
\[
  C_{\cal M}(n,d)\leq 2R_{\cal F}(n,d+1)+2n-2.
\]
In particular, if Restivo's conjecture is true, the latter bound gives
\[
  C_{\cal M}(n,d) = O(n^2),
\]
thus providing a quadratic bound in the size of the set for the minimal length of a synchronizing pair of a finite synchronizing complete code.

In the second part of the paper, we study the dependence of the parameters $R_{\cal L} (n,d)$ and $C_{\cal L} (n,d)$ upon the number of letters $d$ of the considered alphabet, by showing that both the parameters have a low rate of growth.
More precisely, we show that, for the class ${\cal L}$ of finite languages (resp., codes, prefix codes), we have
\[
  R_{\cal L}(n,d)\leq \left\lceil \frac {R_{\cal L}(\lceil\log_2 d\rceil n,2)}{\lfloor\log_2 d\rfloor} \right\rceil,
\]
and, for the class $\cal L$ of finite complete languages (resp., codes, prefix codes), we have
\[
  C_{\cal L}(n,d)\leq \left\lceil \frac{C_{\cal L} (\lceil\log_2(d+1)\rceil n,2)}{\lfloor\log_2(d-1)\rfloor} \right\rceil.
\]
A similar result is obtained also when $\cal L$ is the class of finite (not necessarily complete) languages (resp., codes, prefix codes).

All the latter results were presented with a sketch of the proof in \cite{CDICTCS14,CDMONS14}.

The paper is structured as follows.
In Section \ref{due}, some basic results about complete and synchronizing codes as well as synchronizing automata and \v Cern\'y Conjecture are given.
In Section \ref{tre} we describe our main result.
In Section \ref{quattro}, a study of the dependence of the parameters $R_{\cal L}(n,d)$ and $C_{\cal L}(n,d)$ from the number $d$ of letters of the alphabet is presented.
Finally, in Section \ref{cinque}, some open questions about Restivo Conjecture are formulated.

\section{Preliminaries}
\label{due}
In this section we shortly recall some basic results of the theory of automata and of the theory of codes which will be useful in the sequel and we fix the corresponding notation used in the paper.
The reader can refer to \cite{BP-book,cdel} for more details.

\subsection{Complete and synchronizing sets} Let $A$ be a finite alphabet and let $A^{*}$ be the free monoid of words over the alphabet $A$.
The identity of $A^*$ is called the {\em empty word} and is denoted by $\epsilon$.
The {\em length} of a word $w\in A^*$ is the integer $|w|$ inductively defined by $|\epsilon|= 0$, $|wa| = |w| + 1$, $w\in A^*$, $a\in A$.
Given $w\in A^*$ and $a\in A$, we denote by $|w|_a$ the number of occurrences of the letter $a$ in $w$.
For any finite set of words $W$ we denote by $\ell({W})$ the maximal length of the words of $W$.
The number $\ell({W})$ will be called the {\em size} of $W$.
Given words $u, w\in A^*$, $u$ is said to be a {\em factor} of $w$ if $w = \alpha u \beta,$ for some $\alpha, \beta \in A^*$.
The set of all factors of $w$ is denoted by $\fact (w)$.
Given a set of words $W$, the set of the factors of all the words of $W$ is denoted by $\fact (W)$.
Similarly, given a word $w$, a word $u$ is said to be a {\em prefix} of $w$ if $w = u \beta,$ for some $ \beta \in A^*$.
A set $X$ is said to be {\em prefix} if no word of $X$ is a prefix of another word of $X$.

\begin{definition}
  Let $X$ be a subset of $A^*$.
  A pair of words $(r,s)$ is an {\em $X$-completion} of a word $w$ if $rws\in X^*$.
  A word having an $X$-completion is a {\em completable} word of $X$; conversely, a word with no $X$-completion is an {\em incompletable} word of $X$.
  The set $X$ is {\em complete} if all words of $A^*$ are completable words of $X$; $X$ is {\em incomplete}, otherwise.
\end{definition}

Another crucial notion of this paper is that of synchronizing set.
\begin{definition}
  Let $X$ be a subset of $A^*$.
  A pair $(u,v)\in X^*\times X^*$ is a {\em synchronizing pair} of $X$ if for every $X$-completion $(r,s)$ of $uv$, one has
  \[
    ru,vs\in X^*\,.
  \]
  The set $X$ is {\em synchronizing} if it has a synchronizing pair.
\end{definition}

\begin{example}
  Consider the set
  \[
    X=\{aa,ab,ba,baa,bbb\}
  \]
  on the alphabet $A=\{a,b\}$.
  The pair $(b,aa)$ is a $X$-completion of the word $bbabb$.
  Indeed, one has $b\,bbabb\,aa\in X^*$.
  
  One easily verifies that all words of $A^*$ of length 6 have an $X$-completion.
  On the contrary, the word $v=abbabba$ has no $X$-completion.
  Thus, $v$ is an incompletable word of $X$ of minimal length.
  
  It is not difficult to verify that the pair $(ab,ba)$ is a synchronizing pair of the set $X$.
  Thus, $X$ is a synchronizing set.
\end{example}

The notion of synchronizing pair of a set is strictly related to that of {\em constant}.
A word $c$ of $X^*$ is said to be a {\em constant} of $X$ if, for every $u_1, u_2, u_3, u_4\in A^*$ such that $u_1 c u_2, u_3cu_4 \in X^*$, one has $u_1 c u_4, u_3cu_2 \in X^*$.
The following result holds.
\begin{lemma}
  \label{lemmaConstant}
  Let $X$ be a subset of $A^*$.
  If $(u,v)$ is a synchronizing pair of $X$, then $uv$ is a constant of $X$.
  Conversely, if $c$ is a constant of $X$, then $(c,c)$ is a synchronizing pair of $X$.
\end{lemma}

\subsection{Complete and synchronizing codes}
The notions of complete and synchronizing sets provide a rich structure in the case that the set is a code.
It is worth to shortly describe some fundamental results on such sets.
A set $X$ of words over an alphabet $A$ is said to be a {\em (variable length) code over $A$} if it fulfills the unique factorization property, that is, for every word $u\in X^*$, there exists a unique sequence $x_1, \ldots, x_k$ of words of $X$ such that $u= x_1 \cdots x_k$.
A well-known example of codes is given by all prefix set which are distinct from $\{\epsilon\}$.

The notion of code is strictly related to the one of {\em monomorphism} of free monoids.
Indeed, let $A$ and $B$ be two alphabets. As is well known, a morphism $h:A^* \rightarrow B^*$ is injective if and only if the letters of $A$ have distinct images and $h(A)$ is a code.

In the sequel, a monomorphism $h: A^* \rightarrow B^*$ such that $h(A)$ is a prefix code will be called {\em prefix encoding}.

The notion of complete code is related to that of maximal code.
Indeed, a regular code $X$ is complete if and only if it is maximal (that is, it is not a subset of another code on the same alphabet).
Moreover, a prefix code $Y$ on an alphabet $A$ is complete if and only if any word of $A^*$ is a prefix of a word of $X^*$ (see, \emph{e.g.}, \cite{BP-book}).

\subsection{Synchronizing automata and the \v Cern\'y conjecture}
As usually, by \emph{finite non-deterministic automaton} we mean a 5-tuple
$\A=\langle Q,A,\delta, I, F\rangle$,
where $Q$ is a finite set of elements called {\em states},
$A$ is the \emph{input alphabet},
$\delta:Q\times A\longrightarrow\mathcal P(Q)$ is the {\em transition function},
and $I,F\subseteq Q$ are the sets of initial and terminal states
(here, $\mathcal P(Q)$ denotes the power set of $Q$).

With any automaton $\A$ is naturally associated a directed labelled finite multigraph
$G(\A)=(Q,E)$, where the set $E$ of edges is defined as 
\[E=\{(p,a,q)\in Q\times A\times Q\mid q\in\delta(p,a)\}.\]

However, in this paper, we will consider only automata such that
$I=F=\{1\}$, that is, with a unique initial and final state denoted $1$.
Such an automaton will be simply identified by the 4-tuple
$\A=\langle Q,A,\delta, 1\rangle$.
The language accepted by such an automaton is $L(\A)=X^*$, where $X$ is the set of the labels of the paths in the graph $G(\A)$, with origin and goal in the state 1, but with no intermediate vertex equal to 1.

The canonical extension of the map $\delta$ to the set $Q\times A^{*}$ will be still denoted by $\delta$.
Moreover, if $P$ is a subset of $Q$ and $u$ is a word of $A^{*}$, we denote by $\delta(P,u)$ and $\delta(P,u^{-1})$ the sets:
\[
\delta (P, u) = \{\delta (s, u)\mid s\in P\}, \quad \delta (P, u^{-1}) = \{s\in Q \mid \delta (s, u)\in P\}.
\]
If no ambiguity arises, the sets $\delta (P, u)$ and $\delta (P, u^{-1})$ are denoted $Pu$ and $ P u^{-1},$ respectively.

An automaton $\A=\langle Q,A,\delta, 1\rangle$ is said to be {\em transitive}
if the graph $G(\A)$ is strongly connected.
It is not difficult to verify that any automaton $\A$ is equivalent to a transitive automaton whose graph is the strongly connected component of $G(\A)$ containing the state $1$.
For this reason, in the sequel, we will consider only transitive automata.

An automaton $\A=\langle Q,A,\delta,1\rangle$ is said to be {\em unambiguous} if for all $u,v\in A^*$ there is at most one state $q\in Q$ such that $q\in\delta(1,u)$ and $1\in\delta(q,v)$.
This is equivalent to say that any word of $L(\A)$ is the label of a unique path
of $G(\A)$ with origin and goal in the state 1.

We say that an unambiguous automaton $\A=\langle Q,A,\delta,1\rangle$ is {\em synchronizing} if there exist two words $w_1, w_2\in A^*$ such that
$Qw_1\cap Q{w^{-1}_2}=\{1\}$.

The automaton $\A$ is {\em deterministic} if for all $q\in Q$ and for all $a\in A$, $\card(qa)\leq 1$.

The automaton $\A$ is {\em complete} if for all $u\in A^*$, the set $Qu$ is non-empty.

The properties of automata defined above reflects some properties of the minimal generating set $X$ of the accepted language $X^*$.
Some of them are summarized in the following lemma.

\begin{lemma}
\label{lm-X}\label{lm-synch}\label{lm-XI}
Let $X\subseteq A^*$ be the minimal generating set of $X^*$
(that is, $X \cap X^2X^* = \emptyset$).
\begin{enumerate}
\item
The set $X$ is a regular code if and only if $X^*$ is accepted by an unambiguous automaton
$\A=\langle Q,A,\delta, 1\rangle$.
\item
The set $X$ is a prefix code if and only if $X^*$ is accepted by a deterministic automaton
$\A=\langle Q,A,\delta, 1\rangle$.
\item
The set $X$ is incomplete if and only if $X^*$ is accepted by a 
transitive incomplete automaton $\A=\langle Q,A,\delta, 1\rangle$.
Moreover, in such a case, a word $w\in A^*$ has an $X$-completion
if and only if $Qw\neq\emptyset$.
\item
The set $X$ is a regular synchronizing code if and only if $X^*$ is accepted by a 
transitive synchronizing unambiguous automaton
$\A=\langle Q,A,\delta, 1\rangle$.
Moreover, in such a case, a pair $(u,v)\in X^*\times X^*$ is a synchronizing pair of $X$
if and only if $Qu\cap Qv^{-1}=\{1\}$.
\end{enumerate}
\end{lemma}

As is well known, a deterministic automaton $\A$ is synchronizing if and only if there is a word $u$ such that the set $Qu$ is reduced to a single state.

Such a word is said to be a \emph{synchronizing word} of $\A$.
The following celebrated conjecture has been raised in \cite{Cerny}.
\medskip\\ {\bf \v Cern\'y Conjecture.} {\em Each synchronizing and complete deterministic automaton with $n$ states has a synchronizing word of length $(n-1)^{2}$.} \medskip

Let us recall an important problem related to the \v Cern\'y Conjecture.
Let $G$ be a finite, directed multigraph with all its vertices of the same outdegree.
Then $G$ is said to be {\em aperiodic} if the greatest common divisor of the lengths of all cycles of the graph is $1$.
The graph $G$ is called a {\em AGW-graph} if it is aperiodic and strongly connected.
The reason why such graphs take this name is due to the fact that these structures were first introduced and studied 
 in the context of Symbolic Dynamics by Adler, Goodwyn and Weiss   in \cite{AGW}.

A {\em synchronizing coloring} of $G$ is a labeling of the edges of $G$ that transforms it into a complete, deterministic and synchronizing automaton.
The {\em Road coloring problem} asks for the existence of a synchronizing coloring for every AGW-graph.
In 2007, Trahtman proved the following remarkable result \cite{trah}.

\begin{theorem}
  \label{trat-thm}
  Every AGW-graph has a synchronizing coloring.
\end{theorem}

We recall that by the well known Kraft-McMillan Theorem (see, \emph{e.g.}, \cite{BP-book}),
integers $k_1, \ldots,k_n, d>0$ are the code-word lengths of a maximal (or, equivalently, complete) prefix code over $d$ letters if and only if they satisfy the condition
\begin{equation}
\label{eq:kraft}
\sum_{i=1}^n d^{-k_i}=1.
\end{equation}

We conclude this section with an application of Trahtman Road-coloring Theorem,
which furnishes a characterization of the code-word lengths of finite complete synchronizing codes.

\begin{proposition}
	\label{sec:red-synch-kt}
	Let $k_1, \ldots,k_n, d>0$ be such that
	\[
	\gcd(k_1, k_2,\ldots,k_n)=1\,,\qquad \sum_{i=1}^n d^{-k_i}=1.
	\]
	Then $k_1,\ldots,k_n$ are the code-word lengths of a synchronizing complete prefix code over $d$ letters.
\end{proposition}

\begin{proof}
	Let $A$ be a $d$-letter alphabet.
	By Kraft-McMillan Theorem, there exists a prefix code $X=\{x_1, \ldots, x_n\}$ over $A$ such that, for every $i=1, \ldots, n$, $|x_i|=k_i$.
	Moreover, such a code is maximal and, consequently, complete.
	
	By {Lemma} \ref{lm-X}, $X^*$ is accepted by a complete deterministic automaton ${\cal A}_X$.
	
	Let $G$ be the underlying graph of ${\cal A}_X$, {\em i.e.}, the graph obtained from ${\cal A}_X$ by ripping off all the labels of its edges.
	Since $\gcd(k_1, k_2,\ldots,k_n)=1$, $G$ is an AGW-graph.
	By {Theorem} \ref{trat-thm}, there exists a synchronizing coloring ${\cal A}'$ of $G$.
	Let $L$ be the language recognized by ${\cal A}'$.
	Again by Lemma~\ref{lm-X}, $L=Y^*$ for a suitable prefix complete synchronizing code $Y$.
	Moreover, by construction, one has $Y=\{y_1, \ldots, y_n\}$ with $|y_i|=|x_i|=k_i$ for every $i=1, \ldots, n$, $|y_i|=k_i$. \qed
\end{proof}

\begin{remark}
	It is worth noticing that the code-word lengths of any finite synchronizing complete  code over $d$ letters satisfies both the conditions of Proposition \ref{sec:red-synch-kt}.
	
	Indeed, as a straightforward consequence of Kraft-McMillan Theorem, the second condition is verified by any maximal (or, equivalently, complete) finite prefix code over $d$ letters.
	
	In order to verify the first one, let $X$ be a finite synchronizing complete code, $(u,v)\in X^*\times X^*$ be a synchronizing pair of $X$, $a\in A$ be a letter, and $(r,s)$ be an $X$-completion of the word $uvauv$.
	Then, one has $ruvauvs\in X^*$ and, consequently, $ru,vau,vs\in X^*$.
	One derives that the greatest common divisor $m$ of the code-word lengths of $X$ has to divide $|u|$, $|v|$, $|vau|$ and also $|vau|-|u|-|v|=1$.
	Thus, $m=1$.
\end{remark}
%
%
%

\section{The main result}
\label{tre}
The main result of this paper is related to a problem that was formulated in \cite{Restivo} by Restivo.
Let $\cal L$ be a class of finite languages.
For all $n>0$ we set
\[
  R_{\cal L}(n)=\sup_{d\geq 1}R_{\cal L}(n,d)\,,\quad C_{\cal L}(n)=\sup_{d\geq 1}C_{\cal L}(n,d)\,.
\]
In \cite{Restivo}, it was conjectured that if $\cal F$ is the class of all finite languages, then $R_{\cal F}(n)\leq 2n^2$.If we restrict ourselves to prefix codes, we get
\begin{proposition}
  {{\em (\cite{Restivo})}} Let ${\cal P}$ be the class of finite prefix codes.
  Then
  \[
    R_{\cal P}(n)\leq 2n^2.
  \]
\end{proposition}
However, in the general case, the previous bound was disproved in \cite{Fici}.
A more general and larger counterexample is given in \cite{gusev}.
We can thus state a slightly weaker version of the problem as follows.
\begin{conjecture}
  {(Restivo's Conjecture)} Let $\cal F$ be the class of all finite languages.
  Then $R_{\cal F}(n)=O(n^2)$.
\end{conjecture}

In this context, the main result of this paper is the following.

\begin{proposition}
  \label{main}
  Let ${\cal M}$ be the class of complete finite codes.
  For all $n,d>0$,
  \[
    C_{\cal M}(n,d)\leq 2R_{\cal F}(n,d+1)+2n-2.
  \]
\end{proposition}
Before proving Proposition \ref{main}, it is convenient to discuss some interesting consequences of this result.
First, if Restivo's conjecture is true, we get
\[
  C_{\cal M}(n) = O(n^2).
\]
Moreover, the bound above would be sharp, as we explain below.
Consider the prefix code $X_n=aA^{n-1} \cup bA^{n-2}$ on the alphabet $A=\{a,b\}$.
The minimal automaton accepting $X_n ^*$ has been studied in \cite{A}, where it has been proved that the minimal length of its synchronizing words is $n^2-3n+3$.
  From this, one derives that any synchronizing pair $(w_1,w_2)$ of $X_n$ verifies $|w_1w_2|\geq(n-1)^2$.
In particular, a synchronizing pair of $X_n$ of minimal length is $((ab^{n-2})^{n-1},\epsilon)$.
This provides the lower bound
\[
  {\cal C}_{\cal M}(n, 2) \geq{\cal C}_{\cal P}(n, 2) \geq (n-1)^2,
\]
for the parameter ${\cal C}_{\cal M}(n, 2)$.

It is also worth to do a remark on a recent result by B\'eal and Perrin.
In \cite{BP} (cf.
also \cite{BBP}), it is proved that a synchronizing complete prefix code $X$ with $n$ code-words has a synchronizing word of length $2(n-2)(n-3)+1$.
This result is derived from an upper bound to the length of shortest synchronizing words of synchronizing one-cluster automata.
However, in view of Proposition \ref{main} and Restivo's conjecture, this bound seems of no help in obtaining a good evaluation of the parameter $C_{\cal P}(n,2)$, as one may have $n\simeq 2^{\ell(X)}$.
This suggests that a bound in term of the size of $X$ may be more informative than a bound in terms of the cardinality.

\subsection{Proof of Proposition \ref{main} } Let us now proceed to prove Proposition \ref{main}.
For this purpose, let $X$ be a finite complete synchronizing code over a $d$-letter alphabet $A$ and let $n=\ell (X)$.
Let ${\cal A}_X = \langle Q, A, \delta, 1 \rangle$ be the unambiguous automaton that accepts $X^*$ (see Lemma \ref{lm-X}).
The proof of Proposition \ref{main} is based upon the following lemma.
\begin{lemma}
  \label{main-lm}
  Let $(v_1,v_2)$ be a synchronizing pair of $X$.
  There exist words $w_1, w_2\in A^*$ such that
  \[
    |w_1|,|w_2|\leq R_{\cal F} (n, d+1),\quad Qw_1\subseteq Qv_1,\quad Q{w_2}^{-1}\subseteq Qv_2^{-1}.
  \]
\end{lemma}
Indeed, assume that Lemma \ref{main-lm} holds.
As $X$ is complete, the word $w_1w_2$ has an $X$-completion $(r,s)$.
With no loss of generality, we may suppose that $|r|,|s|\leq n-1$.
Since $(v_1,v_2)$ is a synchronizing pair, in view of {Lemma} \ref{lm-synch}, one has
\[
  Q(rw_1)\cap Q(w_2s)^{-1}\subseteq Qw_1\cap Q{w_2}^{-1}\subseteq Qv_1\cap Q{v_2}^{-1}= \{1\}.
\]
Moreover, the word $rw_1w_2s\in X^*$ is accepted by $\A_X$ and therefore there is a state $q\in Q$ such that $q\in 1rw_1$ and $1\in qw_2s$.
Thus, $q\in Q(rw_1)\cap Q(w_2s)^{-1}\subseteq\{1\}$, that is, $q=1$.
This proves that $rw_1,w_2s\in X^*$ and by {Lemma} \ref{lm-synch} $(rw_1,w_2s)$ is a synchronizing pair of $X$.
Moreover $|rw_1w_2s| \leq 2R_{\cal F} (n, d+1) + 2n -2$.
By the arbitrary choice of the maximal synchronizing code $X$, one derives Proposition \ref{main}.

\medskip

Now, our main goal is to prove Lemma \ref{main-lm}.
For the sake of simplicity, we will prove the existence of the word $w_1$ that fulfills the conditions of Lemma \ref{main-lm} since the proof of the existence of the word $w_2$ can be obtained by using a symmetric construction.
The main tool of this proof is a new automaton we construct below.

Let $(v_1,v_2)$ be a synchronizing pair of $X$.
If $v_1=\epsilon$, the statement is trivially verified by $w_1=v_1$.
Thus we assume $v_1\neq\epsilon$ and set $v_1=ua,$ with $u\in A^*$ and $a\in A$.

Let $a'$ be a symbol not belonging to $A$ and let $A' = A \cup \{a'\}$.
We consider a new automaton ${\A}'=\langle Q,A',\delta',1\rangle$ where the transition map $\delta'$ is defined as follows: for every $q\in Q$ and $a\in A$, $\delta'(q,a)=\delta(q, a)$ and
\begin{equation}
  \label{main-lm-eq:delta'}
  \delta ' (q, a') =
  \begin{cases}
    \delta (q, a)\cup\{1\} & \mbox{if } q\notin\delta (Q, u),\\ \delta (q, a)\setminus\{1\} & \mbox{if } q\in\delta (Q, u).
  \end{cases}
\end{equation}
It is useful to remark that, for all $q\in Q$ and for any word $w\in A^*$, $\delta ' (q, w) = \delta (q, w)$.
It is also useful to remark that, by construction, the automaton ${\cal A}'$ is still transitive.
Let $Y$ be the minimal generating set of the language accepted by $\A'$.
Thus, $L_{{\A}'} = Y^*$ and $Y \cap Y^2 Y^*= \emptyset$.

Now we prove some combinatorial properties of the set $Y$.

\begin{lemma}
  \label{main-lm-1}
  The set $Y$ is incomplete.
\end{lemma}
\begin{proof}
  By (\ref{main-lm-eq:delta'}) one has $\delta'(Q,ua')=\delta(Q,ua)\setminus\{1\}=\delta(Q,v_1)\setminus\{1\}$ and $\delta'(Q,v_2^{-1})=\delta(Q,v_2^{-1})$.
  Taking into account that $(v_1,v_2)$ is a synchronizing pair of $X$, one derives
  \[
    \delta'(Q,ua')\cap\delta'(Q,v_2^{-1})= \delta(Q,v_1)\cap\delta'(Q,v_2^{-1})\setminus\{1\}=\emptyset\,.
  \]
  It follows that $\delta'(Q,ua'v_2)=\emptyset$.
  This equation proves that the automaton $\A$ is not complete.
  Thus, by Lemma~\ref{lm-XI}, $Y$ is an incomplete set.
  \qed
\end{proof}

\begin{lemma}
  \label{main-lm-2}
  It holds that $\ell(Y) \leq \ell(X)$.
\end{lemma}

\begin{proof}
  In order to prove the statement, it is enough to show that, for every $y\in Y,$ there exists $x\in X$ with $|y| \leq |x|$.
  
  Let $y=a_1\cdots a_k\in Y$, with $a_i \in A'$, for $i=1, \ldots, k$.
  Since $Y \cap Y^2 Y^*= \emptyset$, in the graph of $\A'$ there is a path
  \[
    c' = 1 \stackrel{a_1}{\longrightarrow} q_1 \stackrel{a_2}{\longrightarrow}q_2 \stackrel{a_3}{\longrightarrow} \cdots \stackrel{a_{k-1}}{\longrightarrow} q_k \stackrel{a_k}{{\longrightarrow} } 1,
  \]
  where, for every $i=1,\ldots, k,$ $q_i \neq 1$.
  Let us now construct a path $c$ in the graph of ${\cal A}_X$ such that $||c||=x\in X,$ with $|x| \geq |y|$, so completing the proof.
  
  By the definition of $\A'$, any edge $p\stackrel{b}{{\longrightarrow} } q$ of the graph of $\A'$ with $b\neq a'$ is also an edge of the graph of $\A$.
  Moreover, if $p\stackrel{a'}{{\longrightarrow} } q$ is an edge of the graph of $\A'$ with $q\neq 1$, then $p\stackrel{a}{{\longrightarrow} } q$ is an edge of the graph of $\A$.
  Thus, by replacing in $c'$, every transition $q_i \stackrel{a'}{{\longrightarrow} } q_{i+1}$, by $q_i \stackrel{a}{{\longrightarrow} } q_{i+1}$ and deleting the last edge $q_k \stackrel{a_k}{{\longrightarrow} } 1$, we find a path
  \[
    d = 1 \stackrel{b_1}{\longrightarrow} q_1 \stackrel{b_2}{\longrightarrow}q_2 \stackrel{b_3}{\longrightarrow} \cdots \stackrel{b_{k-1}}{\longrightarrow} q_k \stackrel{b_k}{{\longrightarrow} } 1,
  \]
  of the graph of ${\cal A}$.
  Since $\A$ is transitive, one can catenate $d$ with a simple path from $q_k$ to $1$.
  In such a way, we obtain a path $c$ of the graph of $\A$ starting and ending in $1$, with all intermediate states distinct from $1$ and length $\geq k+1$.
  As is well known, as $\A$ is unambiguous, the label $x$ of such a path is a word of the minimal generating set $X$ of $X^*$.
  Since $|x| \geq k+1=|y|$, this completes the proof.\qed
\end{proof}

\begin{lemma}
  \label{main-lm-3}
  Let $v$ be an incompletable word of $Y$ of minimal length.
  There exists a word $w_1\in A^*$ such that
  \[
    |w_1|\leq |v|, \quad Qw_1 \subseteq Qv_1.
  \]
\end{lemma}

\begin{proof}
  Let $v$ be an incompletable word of $Y$ of minimal length, with the number $|v|_{a'}$ as small as possible.
  Then, by {Lemma} \ref{lm-XI}, one has $\delta'(Q,v)=\emptyset$.
  
  The letter $a'$ necessarily occurs in $v$, since by the completeness of $\A$, one has $\delta'(Q,r)=\delta(Q,r)\neq\emptyset$ for all $r\in A^*$.
  Thus, we can write $v=u_1a'u_2$, with $u_1\in A^*$ and $u_2\in A'^*$.
  
  Recall that $v_1 = ua,$ with $u\in A^*$, $a\in A$.
  Let us verify that $\delta(Q,u_1)\subseteq\delta(Q,u)$.
  Indeed, suppose the contrary.
  Then, by (\ref{main-lm-eq:delta'}), one has
  \[
    \delta'(Q,u_1a')=\delta(Q,u_1a)\cup\{1\}=\delta'(Q,u_1a)\cup\{1\}
  \]
  and consequently, $\delta'(Q,u_1au_2)\subseteq\delta'(Q,u_1a'u_2)=\emptyset$.
  Thus, $u_1au_2$ is an incompletable word of $Y$, but this contradicts the minimality of $|v|_{a'}$.
  
  We conclude that $\delta(Q,u_1)\subseteq\delta(Q,u)$ and therefore taking $w_1=u_1a$ and recalling that $v_1=ua$, one has $\delta(Q,w_1)\subseteq\delta(Q,v_1)$ and $|w_1| \leq |v|$.
  The statement follows.\qed
\end{proof}

Let us finally remark that Lemma \ref{main-lm-2} and Lemma \ref{main-lm-3} yield
\[
  |w_1| \leq R_{\cal F} (n, d+1),\quad Qw_1\subseteq Qv_1.
\]
The proof of Lemma \ref{main-lm} is thus complete.

If we restrict ourselves to prefix codes, we obtain a tighter bound.

\begin{proposition}
  \label{prefisso}
  Let ${\cal MP}$ be the class of complete finite prefix codes.
  For all $n,d>0$,
  \[
    C_{\cal MP}(n,d)\leq R_{\cal F}(n,d+1)\,.
  \]
\end{proposition}

\begin{proof}
  Let $X$ be a maximal prefix code.  Then, $X$ is accepted by a complete deterministic automaton $\mathcal A_X$.
  Moreover, $X$ has a synchronizing pair $(v_1,v_2)$ with $v_2=\epsilon$.
  Thus, $Qv_1=Qv_1\cap Qv_2^{-1}=\{1\}$.
  By Lemma \ref{main-lm}, there is a word $w_1\in A^*$ such that
  \[
    |w_1|\leq R_{\cal F}(n,d+1)\,,\quad Qw_1=\{1\}\,.
  \]
  This implies that $w_1\in X^*$ and $(w_1,\epsilon)$ is a synchronizing pair of the prefix code $X$.
  This proves the statement.\qed
\end{proof}

\begin{example}
  \label{ex:det}
  Consider the prefix code
  \[
    X=\{a,baaa,baab,bab,bb\}\,.
  \]
  The automata $\A_X$ and $\A'$ are represented in Figure \ref{figA}.
  One obtains

  \begin{figure}
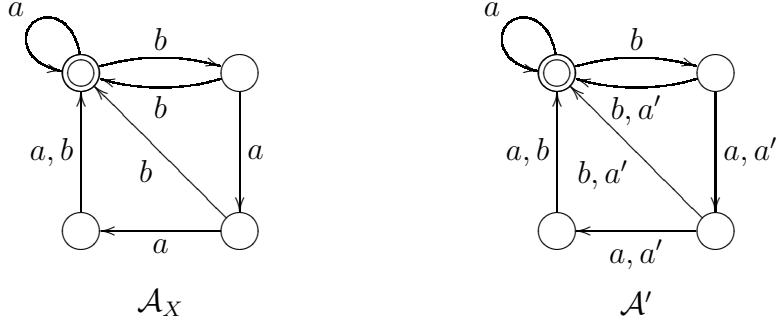

    \label{figA}
    \[
      \def\automamore{\POS(.5,-.45)*{\A_X}\POS(3.5,-.45)*{\A'}} \automa[0,1;1,1;1,0;0,0;3,1;4,1;4,0;3,0]%
      [0,0,a,ul;0,1,b,2;1,0,b,2;1,2,a;2,3,a;2,0,b;3,0,{a,b};%
      4,4,a,ul;4,5,b,2;5,4,{b,a'},2;5,6,{a,a'};6,7,{a,a'};6,4,{b,a'};7,4,{a,b}][][0;4]
    \]
    \caption{Automata of Example \ref{ex:det}}
  \end{figure}
  \begin{multline*}
    Y=\{a,ba',bb,baa',bab,ba'a',ba'b,baaa,baab,\\ baa'a,baa'b,ba'aa,ba'ab,ba'a'a,ba'a'b\},
  \end{multline*}
  so that $\ell(Y)=\ell(X)=4$.
  The word $aaa'$ is $Y$-incompletable and, consequently, $(aaa,\epsilon)$ is a synchronizing pair of the code $X$.
\end{example}

\section{Reduction to the binary case}
\label{quattro}
\label{sec:red}
The aim of this section is to study how much the parameters $R_{\cal L}(n,d)$ and $C_{\cal L}(n,d)$ vary according to the number $d$ of letters of the alphabet.
We start to analyze the parameter $R_{\cal L}(n,d)$.
In the sequel, $B$ denotes the binary alphabet $B=\{a, b\}.$ The following lemma will be useful in the sequel.
It gives an interesting insight on the structure of the completions of words in a complete regular set.
As far as we know, it seems to be a new result.
\begin{lemma}
  \label{sec:red-comp-0}
  Let $Y\subseteq A^*$ be a complete regular set.
  Then any word $w$ of $A^*$ has a $Y$-completion $(y,s)$ with $y\in Y^*$.
\end{lemma}

\begin{proof}
  We define an infinite sequence $((u_n,v_n))_{n\geq 0}$ as follows: $(u_0,v_0)$ is a $Y$-completion of $w$; for all $n>0$, $(u_n,v_n)$ is a $Y$-completion of the word
  \[
    wv_0wv_1\cdots wv_{n-1}w.
  \]
  By Myhill-Nerode Theorem (see, \emph{e.g.}, \cite{BP-book}), $Y^*$ is union of congruence classes of a congruence of finite index $\equiv$\,.
  Thus, one has $u_h\equiv u_k$ for some $h,k$ with $k>h\geq 0$.
  By construction,
  \[
    x=u_kwv_0wv_1\cdots wv_k\in Y^*\quad\mbox{and}\quad z=u_hwv_0wv_1\cdots wv_h\in Y^*.
  \]
  One can write $x=yws$, with $y=u_kwv_0wv_1\cdots wv_h$ and $s=v_{h+1}wv_{h+2}\cdots wv_k$, so that $(y,s)$ is a $Y$-completion of $w$.
  Moreover, one has $y\equiv z$ and, consequently, $y\in Y^*$.
  This concludes the proof.\qed
\end{proof}

\begin{lemma}
  \label{sec:red-comp-1}
  Let $h:A^*\rightarrow B^*$ be a prefix encoding and $Y\subseteq A^*$.
  The set $h(Y)$ is complete if and only if $Y$ and $h(A)$ are complete.
\end{lemma}

\begin{proof}
  $(\Leftarrow)$ Let $w\in B^*$.
  Since $h(A)$ is complete, one has $rws=h(u)\in h(A^*)$, for some $r,s\in B^*$ and $u\in A^*$.
  Since $Y$ is complete, one has $puq\in Y^*$, where $p,q\in A^*$, thus yielding $h(puq)=h(p) rws h(q)\in h(Y^*)$.
  Hence $(h(p)r, s h(q))$ is a $h(Y)$-completion of $w$.
  
  $(\Rightarrow)$ The fact that $h(A)$ is complete follows straightforwardly from the inclusion $B^* \subseteq \fact (h(Y^*)) \subseteq \fact (h(A^*))$.
  
  Let us prove that $Y$ is complete.
  Let $w\in A^*$.
  Since $h(Y)$ is complete, by Lemma \ref{sec:red-comp-0}, one has $h(u) h(w) s= h(v),$ for some $u,v\in Y^*$ and $s\in B^*$.
  Since $h$ is a prefix encoding, one has $v = uwr,$ for some $r\in A^*$.
  The latter implies that $(u, r)$ is a $Y$-completion of $w$.
  \qed
\end{proof}

By encoding a $d$-letter alphabet on a suitable complete binary prefix code one obtains
\begin{proposition}
  \label{sec:red-comp-2}
  Let ${\cal L}$ be the class of finite languages (resp., codes).
  Then
  \begin{equation}
    \label{sec:red-synch-eq:stat-lm-first}
    R_{\cal L}(n,d)\leq \left\lceil \frac {R_{\cal L}(\lceil\log_2 d\rceil n,2)}{\lfloor\log_2 d\rfloor} \right\rceil.
  \end{equation}
\end{proposition}

\begin{proof}
   Let $A$ be a $d$-letter alphabet and let $X$ be a finite incompletable language over $A$ of size $n$.
  Set $m = \lceil\log_2 d\rceil$, $\gamma = 2^{m+1}-d$ and let %
  $k_1, \ldots, k_d$ be the positive integers defined by
  \begin{equation}
    \label{eq:lunghezze}
    k_i=
    \begin{cases}
      m &\mbox{if } i\leq\gamma\,,\\ m+1 &\mbox{if } \gamma<i\leq d\,.
    \end{cases}
  \end{equation}
  One easily checks that 
  \begin{equation}
    \label{eq:completo}
    \sum_{i=1}^dk_i=1.
  \end{equation}
  Thus, by Kraft-McMillan Theorem, $k_1,\ldots,k_d$ are the code-word lengths of a synchronizing prefix code $Y$ over a binary alphabet $B$.
  Moreover, (\ref{eq:completo}) ensures that $Y$ is maximal and, consequently, complete.
  
  Now, let $h:A^*\rightarrow B^*$ be a monomorphism such that $h(A)=Y$.
  Then, for every $a\in A$, we have
  \begin{equation}
    \label{sec:red-comp-eq:1}
    \lfloor \log_2 d\rfloor \leq|h(a)|\leq \lceil \log_2 d \rceil.
  \end{equation}
  By (\ref{sec:red-comp-eq:1}) the size of $h(X)$ is not greater than $n\lceil\log_2 d \rceil$.
  By Lemma \ref{sec:red-comp-1}, since $X$ is incompletable, $h(X)$ is incompletable as well.
  Let $v$ be an incompletable word in $h(X)$ of minimal length.
  Hence we have
  \begin{equation}
    \label{sec:red-comp-eq:1(bis)}
    |v| \leq R_{\cal L}(\lceil\log_2 d\rceil n,2).
  \end{equation}
  Since $Y=h(A)$ is a complete prefix code, the word $v$ is a prefix of a word of $Y^*$.
  Thus, $vs=h(u)$ for some $u\in A^*$ and $s\in B^*$.
   %
  %
  Moreover, taking $u$ of minimal length, one may assume that $u=u'a,$ with $u'\in A^*, a\in A$, and $|h(u')| < |v|$.
  In view of (\ref{sec:red-comp-eq:1}), one derives
  \begin{equation}
    \label{sec:red-comp-eq:}
    |u| \leq \left \lceil \frac{|v| }{\lfloor \log_2 d \rfloor} \right \rceil.
  \end{equation}
  Let us check that $u$ is incompletable in $X$.
  By contradiction, deny.
  Then $r'us'\in X^*$, for some $r',s'\in A^*$.
  Consequently, $h(r'us')=h(r') vs h(s') \in h (X^*)$, thus implying that $v$ is completable in $h (X)$.
  
  Now (\ref{sec:red-synch-eq:stat-lm-first}) easily follows from the latter, (\ref{sec:red-comp-eq:1(bis)}) and (\ref{sec:red-comp-eq:}).
  \qed
\end{proof}

Let us now analyze the parameter $C_{\cal L}(n,d)$.
The following lemma is useful for this purpose.
It is algebraically similar to {Lemma} \ref{sec:red-comp-1}.

\begin{lemma}
  \label{sec:red-synch-1}
  Let $h:A^*\rightarrow B^*$ be a monomorphism and let $Y\subseteq A^*$ be a complete set.
  The set $h(Y)$ is synchronizing if and only if $Y$ and $h(A)$ are synchronizing.
\end{lemma}

\begin{proof}
  $(\Leftarrow)$ By hypothesis and Lemma~\ref{lemmaConstant}, there exists a word $y\in Y^*$ which is a constant of $Y^*$.
  Similarly, there exists a word $h(u)\in h(A^*)$, with $u\in A^*$, which is a constant of $h(A^*)$.
  Since $Y$ is complete, there exist words $r, s \in A^*$ such that $rus\in Y^*$.
  Let $\zeta = h (rus)\in h(Y^*)$.
  Obviously, $\zeta$ is a constant of $h(A^*)$.
  
  Let us prove that $\zeta h(y) \zeta \in h(Y^*)$ is a constant of $ h(Y^*)$.
  For this purpose, let $\alpha_1, \alpha_2, \alpha_3, \alpha_4\in B^*$ be such that $\alpha_1\zeta h(y) \zeta \alpha_2, \ \alpha_3\zeta h(y) \zeta \alpha_4 \in h(Y^*)$.
  Let us prove that $(\alpha_1, \alpha_4)$ and $(\alpha_3, \alpha_2)$ are $h(Y)$-completions of $\zeta h(y) \zeta$.
  By the latter condition and since $\zeta$ is a constant of $h(A^*)$, one has $\alpha _1 \zeta\in h(A^*)$ so that $\alpha _1 \zeta = h(\beta_1)$, for some $\beta _1\in A^*$.
  Similarly, one has $\zeta\alpha _2 = h(\beta_2)$, $\alpha _3\zeta = h(\beta_3)$, $\zeta\alpha _4 = h(\beta_4)$, for some $\beta _2, \beta _3, \beta _4 \in A^*$.
  The previous two conditions now imply $h(\beta_1 y \beta_2), h(\beta_3 y \beta_4)\in h(Y^*)$.
  Since $h$ is an injective map, the latter implies that $\beta_1 y \beta_2, \beta_3 y \beta_4\in Y^*$.
  Since $y$ is a constant of $Y^*$, one thus have $\beta_1 y \beta_4, \beta_3 y \beta_2\in Y^*$ so that $h(\beta_1 y \beta_4), h(\beta_3 y \beta_2)\in h(Y^*)$, so implying that $(\alpha_1, \alpha_4)$ and $(\alpha_3, \alpha_2)$ are $h(Y)$-completions of $\zeta h(y)\zeta$.
  
  $(\Rightarrow)$ Let $(h(y_1), h(y_2))$ be a synchronizing pair of $h(Y)$, with $y_1, y_2\in Y^*$.  One easily proves that $(y_1, y_2) $ is a synchronizing pair of $Y$.
  Indeed, if $ry_1y_2s\in Y^*$, with $r,s\in A^*$, one gets $h(ry_1y_2s)\in h(Y^*)$ which yields $h(ry_1), h(y_2s)\in h(Y^*)$.
  Since $h$ is an injective map, from the latter we get $ry_1, y_2s\in Y^*$.
  Thus $Y$ is a synchronizing set.
  
  Let us prove now that $h(A)$ is a synchronizing set of $B^*$ as well.
  More precisely, let us prove that the pair $(h(y_1), h(y_2)) $ above considered, is a synchronizing pair of $h(A)$.
  For this purpose, let $r, s\in B^*$ such that $r h(y_1)h(y_2)s\in h(A^*)$.
  Hence there exists $t\in A^*$ such that $r h(y_1y_2)s=h(t)$.
  On the other hand, since $Y$ is complete, there exist words $t_1, t_2\in A^*$ such that $t_1 t t_2 \in Y^*$, which implies $h(t_1 t t_2)=h(t_1) r h(y_1) h(y_2) s h(t_2) \in h(Y^*)$.
  Since $(h(y_1), h(y_2))$ is a synchronizing pair of $h(Y^*)$, one derives $h(t_1) r h(y_1), \ h(y_2) s h(t_2) \in h(Y^*)$.
  Thus, one has
  
  \begin{equation}
    \label{sec:red-sybch-eq:1}
    h(t_1),\ h(t_1) r h(y_1),\ r h(y_1)h(y_2)s,\ h(y_2) s h(t_2),\ h(t_2) \in h(A^*).
  \end{equation}
  Taking into account that $h(A)$ is a code and, consequently, there is a unique factorization of the word $h(t_1) r h(y_1)h(y_2) s h(t_2)$ as product of words of $h(A)$, one derives
  \[
    r h(y_1),\ h(y_2) s \in h(A^*)\,.
  \]
  Hence, $(h(y_1), h(y_2))$ is a synchronizing pair of the code $h(A)$.
  This completes the proof.
  \qed
\end{proof}

As an application of the two lemmas above, by encoding a $d$-letter alphabet on a suitable complete binary synchronizing code, one obtains the following result:

\begin{proposition}
  \label{sec:red-synch-2}
  Let $\cal L$ be the class of finite complete languages (resp., codes, prefix codes).
  Then
  \begin{equation}
    \label{sec:red-synch-eq:stat-lm}
    C_{\mathcal L}(n,d)\leq\left\lceil{ \frac{C_{\mathcal L}(\lceil\log_2(d+1)\rceil n,2)}{\lfloor\log_2(d-1)\rfloor} }\right\rceil.
  \end{equation}
\end{proposition}

\begin{proof}
  Let $A$ be a $d$-letter alphabet and let $X$ be a finite complete synchronizing language over $A$ of size $n$.
  
  First, we consider the case that $d$ is not a power of 2.
  Set $m = \lfloor\log_2 d\rfloor$, $\gamma = 2^{m+1}-d$ and let %
  $k_1, \ldots, k_d$ be the positive integers defined by (\ref{eq:lunghezze}).
  One easily checks that both the conditions of Proposition \ref{sec:red-synch-kt} are satisfied.
  Thus, $k_1,\ldots,k_d$ are the code-word lengths of a synchronizing complete prefix code $Y$ over a binary alphabet $B$.
  
  Now, let $h:A^*\rightarrow B^*$ be a monomorphism such that $h(A)=Y$.
  Then (\ref{sec:red-comp-eq:1}) is verified by every $a\in A$, so that the size of $h(X)$ is not greater than $n\lceil\log_2 d \rceil$.
  Since $X$ is a synchronizing and complete set and $Y$ is a synchronizing and complete code, by {Lemma} \ref{sec:red-comp-1} and {Lemma} \ref{sec:red-synch-1}, one has that $h(X)$ is a synchronizing and complete set as well.
  Moreover, if $X$ is a code (resp., a prefix code), then $h(X)$ is a code (resp., a prefix code), too.
  
  Let $(h(u), h(v))$ be a synchronizing pair of $h(X)$, $u,v\in X^*$.
  Hence we have
  \begin{equation}
    \label{sec:red-synch-eq:2}
    |h(uv)| \leq {C}_{\cal L}(n\lceil\log_2 d \rceil,2).
  \end{equation}
  It is easily checked that $(u, v)$ is a synchronizing pair of $X$.
  Indeed, let $ruvs\in X^*$, with $r,s\in A^*$.
  Hence $h(ruvs)\in h(X^*)$ so that $h(ru), h(vs)\in h(X^*)$.
  Since $h$ is an injective mapping, we conclude that $ru, vs\in X^*$.
  
  Hence, by taking account of (\ref{sec:red-comp-eq:1}), (\ref{sec:red-synch-eq:2}), one gets (\ref{sec:red-synch-eq:stat-lm}).

  Finally, let us treat the case where $ d = 2^{m}$.
  Let $k_1, \ldots, k_d$ be the sequence of positive integers defined as: for every $i=1, \ldots, d$,
  \[
    k_i =
    \begin{cases}
      m-1 &\mbox{if } i=1\,,\\ m+1 &\mbox{if } i=2, 3\, , \\ m &\mbox{if } i =4,\ldots, d\,.
    \end{cases}
  \]
  As before, one easily checks that the sequence of lengths $k_1, \ldots, k_d$ defined above satisfy both the conditions of Proposition \ref{sec:red-synch-kt}.
  Thus, $k_1,\ldots,k_d$ are the code-word lengths of a synchronizing complete prefix code $Y$ over a binary alphabet $B$.
  Moreover, for every $a\in A$, we have
  \begin{displaymath}
    \lfloor \log_2(d-1)\rfloor \leq|h(a)|\leq \lceil \log_2 (d+1) \rceil.
  \end{displaymath}
  From that point on, one proceeds by using the same argument of the previous case.
  The proof of the statement is now complete.
  \qed
\end{proof}

A similar bound can be found also in the case where completeness is not required:

\begin{proposition}
  \label{sec:red-synch-3}
  Let $\cal L$ be the class of finite languages (resp.
  codes, prefix codes).
  Then
  \begin{equation}
    \label{sec:red-synch-eq:stat-lm-unif}
    C_{\cal L} (n,d)\leq \left\lceil { \frac{C_{\cal L} (\lceil\log_2(d+1)\rceil n,2)}{\lceil\log_2(d+1)\rceil} } \right\rceil.
  \end{equation}
\end{proposition}

\begin{proof}
  Let $X\subseteq B^m$, with $m\geq 1$ such that $a^m \notin X$ and $a^{m-1}b, ba^{m-1} \in X$.
  It is easily checked that $X$ is a prefix synchronizing code endowed with the synchronizing pair $(ba^{m-1}, a^{m-1}b)$.
  Let $A$ be a $d$-letter alphabet and let $Y$ be a synchronizing set over $A$ such that $\ell(Y)\leq n$.
  We will find a synchronizing pair of $Y$.
  
  We may suppose that $Y\not\subseteq a^*$ since otherwise it has a synchronizing pair $(u,v)$ with $|uv|\leq C_{\cal L}(n,1)\leq C_{\cal L}(n,2)$.
  Let $(y_1, y_2)$ be a synchronizing pair of $Y$.
  With no loss of generality, we may assume that $ab\in \fact (y_1y_2)$, for two suitable distinct letters $a,b$.
  Let $m=\lceil\log_2(d+1)\rceil $ and let us consider the monomorphism $h: A^* \rightarrow B^*$ generated by a bijective mapping between $A$ and a subset of the set $X$ defined above such that
  \[
    h(a)= ba^{m-1}, \quad h(b) = a^{m-1}b.
  \]
  Let us prove that $(h(y_1), h(y_2))$ is a synchronizing pair of $h(Y)$ so that $h(Y)$ is a synchronizing set.
  For this purpose, let $rh(y_1)h(y_2)s \in h(Y^*)$ with $r, s \in B^*$.
  By costruction, we know that $y_1 y_2 = \alpha ab \beta,$ where $\alpha, \beta \in A^*$.
  The latter implies that
  \[
    rh(y_1)h(y_2)s=r h(\alpha ) ba^{m-1} a^{m-1}b h(\beta)s \in X^*\,.
  \]
  Since $(ba^{m-1}, a^{m-1}b)$ is a synchronizing pair of $X$ and $X$ is a uniform length code, from the latter equation one has $r,s \in X^*$ and thus $r = h(r')$ and $s = h(s')$ with $r', s' \in A^*$.
  Hence $r h(y_1) h(y_2)s = h(r'y_1y_2s') \in h(Y^*)$.
  By the injectivity of $h$, one has $r'y_1y_2s'\in Y^*$.
  Since $(y_1, y_2)$ is a synchronizing pair of $Y$, one derives $r' y_1, y_2s' \in Y^*$ and thus $rh(y_1), h(y_2)s\in h(Y^*)$.
  
  Now, using an argument similar to that used in the proof of {Proposition} \ref{sec:red-synch-2} and by remarking that, for every $w\in A^*$, $ |h(w)|= |w| \lceil\log_2(d+1)\rceil$, one proves (\ref{sec:red-synch-eq:stat-lm-unif}).
  \qed
\end{proof}

\section{Conclusions}
\label{cinque}
In this paper we have studied the minimal lengths of incompletable and synchronizing words of a finite set $X$ in terms of the size of $X$.
In particular, we have shown some relations among the parameters $R_\cF(n,d)$ and $C_\cM(n,d)$ bounding, respectively, the minimal lengths of incompletable words in sets of size $n$ and the minimal lengths of synchronizing pairs in maximal codes of size $ n$.

As we have seen, Restivo conjectured a quadratic bound to the minimal length of incompletable words of any finite incompletable set.
However, up to now, such a bound has been found only for prefix codes.
Thus, we may consider the following unanswered questions, most of which may be viewed as weaker versions of Restivo's Conjecture.
We recall that with $\cF$ we have denoted the class of all finite sets.

\begin{enumerate}
  \item Does $R_\cF(n)<\infty$ for all $n$ holds true?
  \item Find a polynomial upper bound to $R_\cF(n)$.
  \item Find a polynomial upper bound to $R_\cF(n,2)$.
  \item Let $\cF_k$ be the class of all $k$-word languages ($k\geq 2$).
  Evaluate $R_{\cF_k}(n)$.
  \item Does $R_{\cF_k}(n)=R_{\cF_k}(n,2)$ holds true?
  \item Let $\mathcal C$ be the class of finite codes.
  Find a polynomial upper bound to $R_\mathcal{C}(n)$.
  \item Let $\mathcal P$ be the class of finite prefix codes.
  Find the exact value of $R_\mathcal{P}(n)$ for all $n$.
\end{enumerate}

\end{document}